\documentclass[11pt]{preprint}

\title{\bfseries Towards Optimal Running Times for Optimal Transport}

\author{Jose Blanchet\footnote{J. Blanchet gratefully acknowledges support from the following NSF grants 1915967, 1820942, 1838576.} \\
\textsf{Stanford University}\\
\texttt{jose.blanchet@stanford.edu}\\
\and
Arun Jambulapati \\
\textsf{Stanford University} \\
\texttt{jmblpati@stanford.edu} \\
\and
\and
Carson Kent \\
\textsf{Stanford University} \\
\texttt{crkent@stanford.edu} \\
\and
Aaron Sidford \\
\textsf{Stanford University} \\
\texttt{sidford@stanford.edu} \\
}

\newcommand{\Rnonneg}{\mathbb{R}_{\geq 0}}

\newcommand{\defeq}{\stackrel{\mathrm{def}}{=}}
\newcommand{\spchain}{\left(M^{(1)}, \ldots, M^{(d)}; F_1, \ldots, F_{d-1}\right)}
\newcommand{\spmem}{\left(M^{(t)}, F_t\right)}

\newtheorem{corollary}{Corollary}

\theoremstyle{definition}
\newtheorem{definition}{Definition}

\date{}

\begin{document}

\maketitle
\thispagestyle{empty}

\begin{abstract}
In this work, we provide faster algorithms for approximating the optimal transport distance, e.g. earth mover's distance, between two discrete probability distributions $\mu, \nu \in \Delta^n$. Given a cost function $C : [n] \times [n] \to \mathbb{R}_{\geq 0}$ where $C(i,j) \leq 1$ quantifies the penalty of transporting a unit of mass from $i$ to $j$, we show how to compute a coupling $X$ between $r$ and $c$ in time $\widetilde{O}\left(n^2 /\epsilon \right)$ whose expected transportation cost is within an additive $\epsilon$ of optimal. This improves upon the previously best known running time for this problem of $\widetilde{O}\left(\text{min}\left\{ n^{9/4}/\epsilon, n^2/\epsilon^2 \right\}\right)$.

We achieve our results by providing reductions from optimal transport to canonical optimization problems for which recent algorithmic efforts have provided nearly-linear time algorithms. Leveraging nearly linear time algorithms for solving packing linear programs and for solving the matrix balancing problem, we obtain two separate proofs of our stated running time. Further, one of our algorithms is easily parallelized and can be implemented with depth $\widetilde{O}(1/\epsilon)$.

Moreover, we show that further algorithmic improvements to our result would be surprising in the sense that any improvement would yield an $o(n^{2.5})$ algorithm for \textit{maximum cardinality bipartite matching}, for which currently the only known algorithms for achieving such a result are based on fast-matrix multiplication.


\end{abstract}

\newpage

\section{Introduction}
\label{sec:intro}
In this paper, we consider the \textit{discrete optimal transportation problem}.
That is, given two vectors $r$ and $c$ in the $n$-dimensional probability simplex $\Delta^n$,
 we seek
  to compute a coupling $X \in \Delta^{n\times n}$ between $r$ and $c$ such that, for a given, non-negative
  cost function $C : [n] \times [n] \to \Rnonneg$ the expected cost with respect to
  $X$ is minimized. Due to \cite{KantorovitchProb}, this problem has a relatively simple expression as a linear program, namely
\begin{equation} \label{eq:Primal}
\begin{gathered}
  \underset{X \in \Uc(r,c)}{\min} \iprod{C}{X}
\text{ where }
  \Uc(r,c) := \left\{ X \in \Rnonneg^{n \times n} : X\bones = r, X^T \bones = c \right\} ~,
\end{gathered}
\end{equation}
$\iprod{\cdot}{\cdot}$ is the element-wise inner product, $X$ denotes our
\textit{coupling}/\textit{transportation plan} between $r$ and $c$, and $C\in
\Rnonneg^{n \times n}$ is our given cost function expressed as a matrix.
In this paper, we focus on computing additive
$\epsilon$-optimal solutions to \eqref{eq:Primal}, i.e.
 $\hat{X} \in \Uc(r,c)$ such that
\begin{equation} \label{eq:epsOptim}
  \iprod{C}{\widehat{X}} \leq \underset{X \in \Uc(r,c)}{\min} \iprod{C}{X} + \epsilon
\end{equation}

The computation of such solutions, both for discrete distributions $r,c$
 and for distributions over more general metric spaces, is playing an
increasing role in varied tasks throughout machine learning and statistics.
Recent applications in unsupervised learning \cite{UL1}, computer
vision \cite{CV1, CV2},
distributionally-robust optimization \cite{DRO3, DRO1, DRO2}, and statistics \cite{OTStats1, OTStats2} all leverage the ability to compute solutions of
\eqref{eq:Primal} or it's continuous analogues. Moreover, these applications
have created a need for fast (nearly-linear\footnote{We consider nearly-linear
time to be any complexity which is of input size $O(n^2)$ after neglecting factors
in $\epsilon$ and logarithms in $n$.} time) algorithms for
\eqref{eq:Primal} in settings
where the cost function $C$ is quite general-- for instance, in the case where $C$ does not
satisfy metric assumptions.

As a consequence, recent efforts in the fields of optimization and machine
learning \cite{CuturiSinkhorn, WeedSinkhorn, BachSGD, SimplerSink,
  BetterSinkAnaly} have focused on establishing nearly-linear time guarantees through the
development of a sequence of new iterative algorithms for \eqref{eq:Primal}.
This has led to a sequence of increasingly sharper complexity
bounds for \eqref{eq:Primal}.

In this paper we shed light on the complexity of \eqref{eq:Primal} by giving a pair of \textit{simple} reductions from optimal transport to canonical problems in
theoretical computer science, namely \emph{packing linear programming} and
\emph{matrix scaling}. Through these reductions we provide new
algorithms for \eqref{eq:Primal} with improved asymptotic running times to
previous methods\footnote{During the final revision process for this work, a
  paper \cite{ConcurrentPaper} offering partially overlapping results
  was published to ArXiv. This concurrent work constitutes an independent
  research effort. The result which is shared by \cite{ConcurrentPaper} and this
  work is the serial, randomized running time for \eqref{eq:Primal} that is
  obtained in Sections \ref{seq:packing} and \ref{sec:box-constrained} of this
  paper and Theorem 2 of \cite{ConcurrentPaper}. Indeed, a reduction to
  packing LPs which is similar to the one given in Section \ref{seq:packing}
  appears in \cite{ConcurrentPaper}. \cite{ConcurrentPaper} also appeals to further results concerning packing LP solvers and an additional reduction from \eqref{eq:Primal}
  to mixed packing and covering LPs in order to provide deterministic and
  parallel running times for \eqref{eq:Primal} which do not appear in this
  paper-- see Theorem 2 in \cite{ConcurrentPaper}.
  To preserve the independence of these two works following their appearance on
  ArXiv, only edits for clarity were made to the original version of this
  manuscript during preparation-- with one notable exception. The parallel
  complexity for \eqref{eq:Primal} which appears
  in Section~\ref{sec:box-constrained} was added to highlight the difference between
  the reduction of Section~\ref{sec:box-constrained} and the reductions obtained
  in \cite{ConcurrentPaper}-- indeed, in the case of parallel, randomized
  running time, the result of Section~\ref{sec:box-constrained} improves upon 
  \cite{ConcurrentPaper} by a factor of $1/\epsilon$.}\addtocounter{footnote}{-1}\addtocounter{Hfootnote}{-1}. Moreover, we show that these running times cannot be further improved without a major breakthrough in algorithmic graph theory.


\subsection{Contributions and Overview}
The contribution of this paper is two-fold. First, we exhibit two separate
algorithms for computing an $\epsilon$-approximate solution to \eqref{eq:Primal}
in $\widetilde{O}\left(n^2\norm{C}_{\max}/\epsilon \right)$ time.
Throughout, we use  $\widetilde{O}$ to hide logarithmic factors in $n$ and
$\epsilon$ and use $\norm{C}_{\max}$ to denote the largest entry of
$C$. This improves upon the following previous best known complexity for this problem\footnotemark
\begin{equation*}
\widetilde{O}\left(\text{min}\left\{ \frac{n^{9/4}\sqrt{\norm{C}_{\max}}}{\epsilon},
\frac{n^2 \norm{C}_{\max}}{\epsilon^2} \right\}\right)
\hspace{0.1in} \text{\cite{BetterSinkAnaly}}
\end{equation*}
Additionally, one of our algorithms achieves
$\widetilde{O}\left(\norm{C}_{\max}/\epsilon\right)$ parallel depth-- an
improvement on the previously best known parallel depth by a factor of $1/\epsilon$.

These algorithms are derived via black-box reductions to two canonical problems in
theoretical computer science which can be solved using powerful iterative
methods. The first of our reductions is to a standard packing linear program
and the second of our reductions is to the so-called matrix scaling problem.
\begin{definition}[Packing Linear Program]\label{def:packLP}
  A packing linear program is a linear program of the form
\begin{equation}\label{eq:PackLP}
  V_* = \underset{x \in \Rnonneg^l}{\max} \left\{d^Tx : Ax \leq b \right\}
\end{equation}
for $b \in \Rnonneg^m$, $d\in \Rnonneg^l$, and $A \in \Rnonneg^{m \times l}$. We
say that $x_{\epsilon} \in \Rnonneg^m$ is an $\epsilon$-approximate solution for
\eqref{eq:PackLP} if $Ax_\epsilon \leq b$ and $d^Tx_\epsilon \geq (1-
\epsilon)V_*$.
\end{definition}
\begin{definition}[Matrix scaling]\label{def:MatScal}
  Let $A$ be a non-negative matrix and $r,c \in \Rnonneg^n$ be vectors such that
  $\sum_{i=1}^nr_i = \sum_{i=1}^n c_i$ and $\norm{A}_{\max},\norm{r}_{\max},
  \norm{c}_{\max} \leq 1$. Two non-negative diagonal matrices $X,Y$
  are said to \emph{$(r,c)$-scale $A$} if the matrix $B = XAY$ satisfies $B\bones = r$ and
  $B^T\bones = c$.

  If, instead, $\norm{B\bones - r}_1 + \norm{B^T\bones - c}_1 \leq
  \epsilon$ we say that \emph{$X,Y$ $\epsilon-$approximately $(r,c)$-scale $A$}.
  The \emph{matrix scaling problem} is to compute non-negative diagonal matrices $X,Y$ that $\epsilon$-approximately $(r,c)$-scale $A$, provided such matrices exist.
\end{definition}

The techniques used to achieve these reductions are relatively standard; the primary benefit of our result,
beyond the gain of a $O\left(1 / \epsilon\right)$ in complexity, is that
it clarifies the optimal transport problem and exposes its
relationship to two recent algorithmic breakthroughs in theoretical computer
science.

The second contribution of this paper is the demonstration that the above
running time of $\widetilde{O}\left(n^2 / \epsilon \right)$ is
unimprovable by the approach of this paper or others
\cite{BachSGD, WeedSinkhorn, BetterSinkAnaly, MoreC1, MoreC2, MoreC3} for
solving \eqref{eq:Primal}, barring a major breakthrough on a long-standing open problem in combinatorial optimization. More formally, we show that further running time improvements beyond those achieved in this paper would be surprising as they would yield running times for maximum cardinality bipartite matching, which currently are only known to be achievable using fast matrix multiplication.
Given the significant recent attention optimal transport has received
in machine learning and statistics, this
hardness result contributes significantly to clarifying why further
algorithmic improvements for \eqref{eq:Primal} are difficult and additional
problem assumptions may be need to obtain better performance.

As a road-map for the reader, after covering previous work in Section~\ref{sec:prev} and preliminaries in Section~\ref{sec:prelim}, in Section~\ref{seq:packing} we give a reduction from \eqref{eq:Primal} to
a packing linear program (LP) and then show how a recently-developed fast solver for
packing LPs \cite{OrechPacking} can be applied to yield our desired sequential run-time.
In Section~\ref{sec:box-constrained} we give a reduction from \eqref{eq:Primal} to
matrix scaling and then provide our second algorithm, which obtains both our
stated run-time and stated parallel depth. The surprising fact that we can
recover the same overall complexity via these very
different approaches then motivates Section \ref{sec:bimatch} where we prove
our hardness reduction to maximum cardinality bipartite matching.


\section{Previous Work}
\label{sec:prev}
In this paper, we focus on the case of obtaining
nearly-linear running time results for \eqref{eq:Primal}. While we could consider solving \eqref{eq:Primal} as a general linear program,
any approaches involving the fastest known methods (e.g \cite{SidLeeSolv} via
Laplacian system solvers or \cite{SidInvMain} for generic solvers) would be
insufficient for our stated goal since they currently have running time at least $\Omega(n^{2.5})$ for \eqref{eq:Primal}.

Outside of such generic solvers and within the scope of previous algorithms which
achieve nearly-linear running time (or better) for
\eqref{eq:Primal}, contemporary literature comprises two veins. The first vein,
encompasses those algorithms which impose further conditions on the costs of \eqref{eq:Primal} in
order to create a fast computational method for a more restricted subclass of
applications. Examples in this line of work are \cite{MoreC1, MoreC2, MoreC3},
where nearly-linear run-times are obtained, but at the expense of assuming that the cost
matrix $C$ is induced by a metric--
or, in the latter case, by a low dimensional $l_p$ metric. For the purposes of
this paper, we will only make positivity/boundedness
assumptions on our costs (as metric or related assumptions on $C$ can often be
violated in practice). Thus, this line of inquiry is less relevant for our efforts.

The second vein of results, however, is more directly related to the algorithm that we
will present in Section \ref{sec:box-constrained} and stems from the use of
entropy-regularization to solve \eqref{eq:Primal}. Beginning with the work of
\cite{CuturiSinkhorn}, this line of research \cite{BachSGD, WeedSinkhorn,
  SimplerSink, BetterSinkAnaly} essentially centers around applying a particular
 iterative technique, such as alternating minimiziation (Sinkhorn/RAS) or an
 accelerated first order method (APDAGD), to solve the dual of an entropy-regularized version
of \eqref{eq:Primal}. As shown in Table \ref{tbl:complex}, this leads to different
approaches for solving \eqref{eq:Primal} in nearly-linear time. It is worth noting that the procedure which appears in
  Section \ref{sec:box-constrained} is tangentially alluded to in \cite{BetterSinkAnaly}, but no derivation or
  concrete running times were given.

\begin{table}
  \caption{\textbf{Running times for computing $\epsilon$-optimal solutions of
      \eqref{eq:Primal}:}  In the table, $\widetilde{O}$ hides polylogarithmic
    factors in $\epsilon, n$. All results except for the interior point method
    also explicitly hide linear dependence on the norm of the cost matrix $\norm{C}_\infty$} \label{tbl:complex}
  \centering
  \begin{tabular}{|l | c | r|}
    \hline
    Algorithm & Running Time & Paper \\ \hline
    Interior Point & $\widetilde{O}\left( n^{2.5} \right)$ & \cite{SidLeeSolv}\\
    Sinkhorn/RAS & $\widetilde{O}\left( \frac{n^2}{\epsilon^2} \right)$ & \cite{BetterSinkAnaly}\\
    APDAGD & $\widetilde{O}\left(\text{min}\left\{ \frac{n^{9/4}}{\epsilon},
\frac{n^2}{\epsilon^2} \right\}\right)$ & \cite{BetterSinkAnaly}\\
    \hline
    Box-constrained Newton and Packing LP reduction &
                                                      $\widetilde{O}\left(\frac{n^2}{\epsilon}\right)$
                             & \textbf{This paper}\\ \hline
  \end{tabular}
\end{table}


\section{Preliminaries}
\label{sec:prelim}
In this section, we define notation and several, canonical assumptions
concerning \eqref{eq:Primal} that will be relevant for the subsequent reductions.

First, we denote the set of non-negative real numbers by $\Rnonneg$, the set of
integers $\{1,\ldots, n\}$ by $[n]$, and the $n$ dimensional probability simplex
by $\Delta^n = \{x \in \Rnonneg : \sum_{i \in [n]} x_i = 1\}$. Correspondingly, let
$\Delta^{n\times n}= \{x \in \Rnonneg^{n \times n} : \bones^T X \bones = 1\}$ where
$\bones$ is the all ones vector. Given a set $S \subseteq [n]$ and $r \in
\Delta^n$ define $r_{|S}$ to be the conditional distribution induced
by $r$ given $S$. Denote the product
distribution of $r,c \in \Delta^n$ by $r \otimes c \in \Delta^{n \times n}$.

For $A \in \R^{n\times n}$, we define $\norm{A}_{\max}$ to be maximum modulus of
any element of $A$. Further, denote the entry-wise exponential of $A$ by
$e^A$ and for $A \in \Rnonneg^{n\times n}$ define
\begin{equation*}
  H(A) = -\sum_{i,j=1}^n A_{i,j}\left(  \log A_{i,j} - 1\right)
\end{equation*}
to be the (entry-wise) matrix entropy. For two matrices $A,B \in \R^{n\times n}$
we denote the Frobenius inner product by $\iprod{A}{B} = \sum_{i,j\in
  [n]}A_{i,j}B_{i,j}$.

We will refer to the linear program \eqref{eq:Primal}
as the optimal transport problem, Kantorovitch problem, or primal. As is standard, the cost matrix $C \in \Rnonneg^{n\times n}$ has also been assumed to be
non-negative and the marginals have been taken to be strictly
positive ($r, c > 0$). Note,
while we have implicitly assumed that the marginals $r,c \in \Delta^n$ have the same
dimension, this has been done for the sake of exposition and the complexities
will suitably generalize for $r$ and $c$ of differing dimensions-- i.e. our
running times will become $\widetilde{O}\left(mn / \epsilon \right)$ for
$r$ of dimension $m$ and $c$ of dimension $n$.


\section{Solving by Packing LP Algorithms}\label{seq:packing}
In this section, we give a procedure for computing an $\epsilon$-optimal
solution to the optimal transport problem in $\widetilde{O}\left(n^2
  \norm{C}_{\max} / \epsilon \right)$ time.
To obtain our reduction, consider solving the linear program:
\begin{equation}\label{eq:LPForm}
\begin{split}
  \underset{X \in \Kc(r,c)}{\max} \iprod{B}{X}&  \\[0.5em]
  \Kc(r,c) := \left\{ X \in \R^{n\times n}_+ : X\bones \leq r, X^T \bones \leq c
  \right \}\hspace{2em} & \hspace{2em} B := \norm{C}_{\max}\bones \bones^T - C
\end{split}
\end{equation}
In other words, we turn the minimization problem \eqref{eq:Primal} into a
maximization problem by reversing the sign of $C$ while adding
a constant of $\norm{C}_{\max}$ to the constraint matrix to keep the new cost
matrix, $B$, non-negative. This allows us to just solve
under upper bound constraints, rather than both upper and lower bound constraints, on the row and column
sums of $X$. Indeed, the new objective encourages using $X$ to make the row and column constraints tight while
still minimizing the original cost. Furthermore, since $B$ is an entry-wise,
uniform perturbation of $C$ by $\norm{C}_{\max}$, \eqref{eq:LPForm} will
maintain the same set of optimal solutions as \eqref{eq:Primal} while only perturbing
the objective function by an additive $\norm{C}_{\max}$ term-- since $\iprod{X}{\bones \bones^T} = \bones^TX\bones = 1$.

Formally, we first show how to round solutions of \eqref{eq:LPForm} to
solutions of \eqref{eq:Primal}.
\begin{lemma} \label{lem:LPRounding}
  Suppose $X\in \Rnonneg^{n\times n}$ satisfies $X\bones \leq r$ and $X^T \bones
  \leq c$. Then, there exists a matrix $D\in \Rnonneg^{n\times n}$ (which can be
  trivially computed in $O(n^2)$ time) such that $Y = X+D$ satisfies $Y\bones =
  r$ and $Y^T\bones = c$.
\end{lemma}
\begin{proof}
  Define $e_r := r - X\bones$ and $e_c := c - X^T\bones$ and observe $e_r,e_c
  \geq 0$ coordinate-wise and that
  \begin{equation*}
  \norm{e_r}_1 = \bones^T(r - X\bones) = 1 - \bones^TX\bones = (c^T -
    \bones^TX)\bones = \norm{e_c}_1
  \end{equation*}
  Hence, set $D := \frac{1}{\norm{e_c}_1} e_r e_c^T$ where, by convention, $D = 0$ if $\norm{e_c}_1 = 0$.

  It is easy to verify that if $\norm{e_c}_1 = 0$, then $Y=X+D$ has the
  prescribed marginals (row and column sums). Thus, assume that $\norm{e_c}_1
  \neq 0$. Then,
  \begin{equation*}
  Y\bones = \left( X + \frac{1}{\norm{e_c}_1} e_r e_c^T \right) \bones =
    X\bones + e_r = r
  \end{equation*}
  and, similarly, $Y^T\bones = c$.
\end{proof}
Using this lemma, the main result quickly follows:
\begin{theorem}\label{thm:LPred}
  Suppose there exists an oracle $\mathcal{O}$ which computes an
  $\epsilon'$-approximate solution (see Definition \ref{def:packLP}) to the
  packing LP \eqref{eq:LPForm} in time $O(\mathcal{T}\left(m, l,
    1/\epsilon'\right))$.

   Then, there is an
   algorithm which computes an $\epsilon$-approximate solution to the optimal transport problem
   \eqref{eq:Primal} in time
   \begin{equation*}
  O\left(n^2 + \mathcal{T}\left(n, n, \frac{\norm{C}_{\max}}{\epsilon}\right) \right)
   \end{equation*}
\end{theorem}
\begin{proof}
  Let $X_{\epsilon'}$ be the $\epsilon'$-approximate solution obtained by running
  $\mathcal{O}$ on \eqref{eq:LPForm} with approximation parameter $\epsilon' =
  \epsilon/\norm{C}_{\max}$. By Lemma \ref{lem:LPRounding}, we can
  compute a $D \in \Rnonneg^{n\times n}$ in $O(n^2)$ time such that $Y = X_{\epsilon'}+ D$ is feasible for
  \eqref{eq:Primal}. Hence, denoting the optimal solution to the original transportation problem
  \eqref{eq:Primal} by $X_*$, we have
  \begin{equation*}
  \iprod{B}{Y} \geq \iprod{B}{X_{\epsilon'}} \geq (1-\epsilon')\iprod{B}{X_*}
  \end{equation*}
  where we have used the definition of $\epsilon'$-optimality for
  $X_{\epsilon'}$ and the fact that $Y \geq X_{\epsilon '}$ entry-wise.
  Expanding this relationship in $B$ and using the fact that $\bones^T Y \bones
  = 1 $ and $ \bones^T X_* \bones = 1$, we obtain
  \begin{equation*}
  \norm{C}_{\max} - \iprod{C}{Y} \geq \norm{C}_{\max} - \iprod{C}{X_*} -\epsilon'\iprod{B}{X_*}
  \end{equation*}
  Upon rearrangement this yields
  \begin{equation*}
  \iprod{C}{Y} \leq \iprod{C}{X_*} +\epsilon'\iprod{B}{X_*}
  \end{equation*}
  As $\norm{B}_{\max} \leq \norm{C}_{\max}$ and $\epsilon' = \epsilon/\norm{C}_{\max}$, H{\"o}lder's
  inequality implies that
  \begin{equation*}
  \iprod{C}{Y} \leq \iprod{C}{X_*} +\epsilon
  \end{equation*}
  Hence, $Y$ is an $\epsilon$-approximate solution of the optimal transportation
  problem \eqref{eq:Primal}. Moreover, it quickly follows that the total time of
  this procedure is $O\left(n^2 + \mathcal{T}\left(n, n, \norm{C}_{\max}/\epsilon\right) \right)$
\end{proof}
Using this reduction, we can now obtain our desired run-time for
\eqref{eq:Primal},
simply by solving \eqref{eq:LPForm} using the current best packing algorithm.
\begin{theorem}[See \cite{OrechPacking}]\label{thm:packLP}
  Given a packing linear program \eqref{eq:PackLP}, there exists an algorithm that computes an $\epsilon$-approximate
solution to \eqref{eq:PackLP} in time $\widetilde{O}\left(m + l + \emph{nnz}(A)/\epsilon
\right)$ with high probability.
\end{theorem}
With Theorem \ref{thm:packLP} providing the oracle in Theorem
\ref{thm:LPred}, we immediately obtain the following corollary
\begin{corollary}\label{cor:LPfinalRT}
  There exists an algorithm which computes an $\epsilon$-approximate
  solution to the optimal transport problem
   \eqref{eq:Primal} with high probability in time $\widetilde{O}\left(n^2\norm{C}_{\max}/\epsilon \right)$.
\end{corollary}


\section{Solving by Matrix Scaling and Box-constrained Newton}
\label{sec:box-constrained}
In this section, we give a different procedure for computing an $\epsilon$-optimal
solution to the optimal transport problem in time $\widetilde{O}\left(n^2 \norm{C}_\infty/\epsilon \right)$.
The advantage of this new approach is that it not
only obtains the currently-best known sequential run-time, but it also achieves the
fastest known parallel complexity for solving \eqref{eq:Primal} while preserving
total work.

As a first step, we will note the following reduction to the matrix scaling
problem which appears in
prior work \cite{CuturiSinkhorn, WeedSinkhorn, BetterSinkAnaly}. The optimal transport problem naturally yields an entropy-regularized version
\begin{equation}\label{eq:RegPrimal}
  \underset{X \in \,\, \Uc(r,c)}{\min} \iprod{C}{X} - \eta H(X)
\end{equation}
whose optimal value of \eqref{eq:RegPrimal} is called the Sinkhorn cost \cite{CuturiSinkhorn}. The
namesake refers to the fact that the dual of \eqref{eq:RegPrimal} is equivalent to the problem
\begin{equation}\label{eq:MatScal}
  \min_{x,y \in \R^n} \psi(x, y)
\defeq \bones^T B_{C/\eta}(x,y) \bones - r^Tx - c^Ty
\text{ where }
\left(  B_{C/\eta}(x,y)\right)_{ij} \defeq e^{x_i + y_j - C_{ij}/\eta}
\end{equation}
More generally, we will write
\begin{equation}\label{eq:MatScalA}
  \min_{x,y \in \R^n} \psi_{A,r,c}(x, y) \defeq \bones^T M_{A}(x,y) \bones - r^Tx - c^Ty
\text{ where }
\left(M_{A}(x,y) \right)_{ij} \defeq A_{ij}e^{x_i + y_j}
\end{equation}
for any non-negative matrix $A \in \R^{n\times n}$ and positive vectors $r,c \in
R_*^n$. An optimal solution of \eqref{eq:MatScalA} gives diagonal
matrices which $(r,c)$-scale $A$.

It is known that solving \eqref{eq:MatScal} is sufficient to solve the
optimal transport problem in the following sense.
\begin{lemma}[See proof of Theorem 1 in \cite{WeedSinkhorn}]\label{lem:2}
Let $\widehat{x}, \widehat{y}$ be solutions which satisfy
\begin{equation*}
  \norm{B_{C/\eta}(\widehat{x}, \widehat{y}) \bones - r}_1 + \norm{B_{C/\eta}(\widehat{x},
      \widehat{y})^T \bones - c}_1 \leq \epsilon
\end{equation*}
  i.e. $\norm{\nabla \psi(\widehat{x},\widehat{y})}_1 \leq \epsilon$. Then,
  there exists a projection $\widehat{X}$ of $B_{C/\eta}(\widehat{x},
  \widehat{y})$ onto $\Uc(r,c)$ that can be computed in linear-time and work
  (i.e. $O(n^2)$) and $\widetilde{O}(1)$ depth such that
  \begin{equation*}
  \iprod{C}{\widehat{X}} \leq \min_{X\in \Uc(r,c)}\iprod{C}{X} + 2 \eta\log n
    +4 \epsilon \norm{C}_{\infty}
  \end{equation*}
\end{lemma}
Moreover, using Lemma \ref{lem:2} and the following fact, the main reduction of this section is almost immediate.
\begin{lemma}\label{lem:1}
 Given an instance of \eqref{eq:Primal}, there exist a pair of modified, input distributions $\widetilde{r},\widetilde{c}$ such that $\widetilde{r}_i,
 \widetilde{c}_i \geq \frac{\epsilon}{2\norm{C}_{\infty} n}$ for all $i \in [n]$
 and the solution
\begin{equation} \label{eq:PrimalMod}
  \widetilde{X}_* = \underset{X \in \Uc(\widetilde{r},\widetilde{c})}{\min} \iprod{C}{X}
\end{equation}
can be extended to an $\epsilon$-approximate solution $\widehat{X}$ of \eqref{eq:Primal} in
$O\left(n^2\right)$ time/work and $\widetilde{O}(1)$ depth.
\end{lemma}
\begin{proof}
  Let
  \begin{equation*}
  S_r = \left\{ i \in [n] : r_i \geq \frac{\epsilon}{2\norm{C}_{\infty} n}
    \right\} \hspace{1.5em}\text{and}\hspace{1.5em} S_c = \left\{ i \in [n] : c_i \geq \frac{\epsilon}{2\norm{C}_{\infty} n}
    \right\}
  \end{equation*}
  and set $\widetilde{r}$ and $\widetilde{c}$ to be the corresponding marginal
  distributions of $r_{|S_r} \otimes c_{|S_c} \in \Delta^{n\times n}$. Let
  $\widetilde{X}_*$ be the solution of \eqref{eq:PrimalMod} for such
  marginals $\widetilde{r},\widetilde{c}$, denote
  \begin{equation*}
  \mu = \sum_{i\in S_r, j \in S_c} r_ic_j \leq 1
  \end{equation*}
  and set $E = S_r \times S_C \in [n] \times [n]$.
  For the optimal solution $X_*$ of \eqref{eq:Primal} with marginals $r,c$
  and let $X_*^{E}$ be the
  distribution induced by conditioning $X_*$ on the set $E$.

  The optimality of $\widetilde{X}_*$ implies that
  \begin{equation*}
  \iprod{C}{\widetilde{X}_*} \leq \iprod{C}{X_*^{E}} \leq
    \frac{1}{\mu}\iprod{C}{X_*}
  \end{equation*}
  Further, if we let $\widehat{X}$ be the coupling such that
  \begin{equation*}
  \widehat{X}_{ij} = \begin{cases} \mu \widetilde{X}_{ij} & \text{if } i \in S_r,
      j\in S_c \\ r_ic_j & \text{otherwise }\end{cases}
  \end{equation*}
  it is easy to see that $\hat{X}$ has marginals $r$ and $c$ and, by
  construction of $S_r$ and $S_c$, satisfies
  \begin{equation*}
  \iprod{C}{\widehat{X}} \leq \mu\iprod{C}{\widetilde{X}_*} + \epsilon
    \leq \iprod{C}{X_*} + \epsilon
  \end{equation*}
  Clearly, $\widetilde{r}, \widetilde{c}$ and $
  \hat{X}$ can be constructed in $O(n^2)$ time/work and $\widetilde{O}(1)$ depth.
\end{proof}
 \begin{theorem}\label{thm:matScalRed}
   Suppose there exists an oracle $\mathcal{O}$ which computes an $\epsilon'$-approximate
   solution (see Definition \ref{def:MatScal}) to the matrix
   scaling problem in time $O\left(\mathcal{T}\left(n, 1 /\epsilon',
       \nu, \xi \right)\right)$ where $\nu = \max_{i,j} 1/A_{ij}$,
   $\xi= \max_{i \in [n]} \left( 1/\min(r_i,c_i)\right)$, and we let
   $\mathcal{T}\left(n, 1/\epsilon', \nu, \xi \right) = \infty$ when
   $\nu = \infty$ or $\xi = \infty$.
   Then, there is an
   algorithm which computes an $\epsilon$-approximate solution to the optimal transport problem
   \eqref{eq:Primal} in time
   \begin{equation*}
  O\left(n^2 + \mathcal{T}\left(n, \frac{16\norm{C}_\infty}{\epsilon},
      n^{8\norm{C}_\infty/\epsilon}, \frac{4 \norm{C}_\infty n}{\epsilon}\right)
  \right)
   \end{equation*}
 \end{theorem}
\begin{proof} 
  By Lemma \ref{lem:1}, we can assume without loss of generality that $\xi \leq
  \left(4 \norm{C}_\infty n\right) / \epsilon$.
  Set $\eta = \epsilon / (4\log n)$. From \cite{LinialMatrixScaling} and the fact that $e^{- C_{i,j}/\eta},
  r_i,c_i > 0$ we know that $e^{-C/\eta}$ is $(r,c)$-scalable. Thus, by running
  $\mathcal{O}$ on the matrix $e^{-C/\eta}$ with $\epsilon' =
  8\norm{C}_{\infty}/\epsilon$, we can produce an approximate $(r,c)$-scaling
  $B = Xe^{-C/\eta}Y$ such that
  \begin{equation*}
  \norm{B\bones - r}_1 + \norm{B^T \bones - c}_1 \leq \epsilon'
  \end{equation*}
  By Lemma \ref{lem:2}, this scaling can be rounded in $O(n^2)$ time to produce a $\hat{X}$ with
  \begin{equation*}
  \iprod{C}{\widehat{X}} \leq \min_{X\in \Uc(r,c)}\iprod{C}{X} + 2\eta\log
      n+4 \epsilon' \norm{C}_{\infty} \leq \min_{X\in
      \Uc(r,c)}\iprod{C}{X} + \epsilon
  \end{equation*}
  Since
  \begin{equation*}
  \nu = \max_{i,j} \,\, \frac{1}{e^{-C_{ij} / \eta}} \leq \exp\left(\frac{4
        \norm{C}_\infty\log n }{\epsilon}\right) = n^{4\norm{C}_\infty/\epsilon}
  \end{equation*}
  It follows that this procedure takes
  \begin{equation*}
  O\left(n^2 + \mathcal{T}\left(n, \frac{8\norm{C}_\infty}{\epsilon},
        n^{4\norm{C}_\infty/\epsilon}\right) \right)
  \end{equation*}
  total time.
\end{proof}
 \begin{corollary}\label{cor:parMatScalRed}
   Suppose there exists an oracle $\mathcal{O}$ which computes an $\epsilon'$-approximate
   solution to the matrix
   scaling problem in parallel in $O\left(\mathcal{T}_w\left(n, 1 /\epsilon',
       \nu, \xi \right)\right)$ total work and $\widetilde{O}\left(\mathcal{T}_d\left(n, 1 /\epsilon',
       \nu, \xi \right)\right)$ depth.
   Then, there is an
   algorithm which computes an $\epsilon$-approximate solution to the optimal transport problem
   \eqref{eq:Primal} in
   \begin{equation*}
  O\left(n^2 + \mathcal{T}_w\left(n, \frac{16\norm{C}_\infty}{\epsilon},
      n^{8\norm{C}_\infty/\epsilon}, \frac{4 \norm{C}_\infty n}{\epsilon}\right)
  \right)
   \end{equation*}
  work and
   \begin{equation*}
  \widetilde{O}\left(\mathcal{T}_d\left(n, \frac{16\norm{C}_\infty}{\epsilon},
      n^{8\norm{C}_\infty/\epsilon}, \frac{4 \norm{C}_\infty n}{\epsilon}\right)
  \right)
   \end{equation*}
depth.
 \end{corollary}
Given this reduction between matrix scaling and optimal transport, it remains
for us to provide concrete bounds for $\mathcal{T}\left(n, 1/\epsilon',
  \nu, \xi\right)$ in order to show our desired run-time. To this
end, consider the following guarantee given by a currently best algorithm for the matrix scaling problem\footnote{It should be remarked that
  similar results to \cite{CohenBoxCons} were obtained independently by \cite{OtherScaling}. We
focus on the guarantee stated in \cite{CohenBoxCons} since it is more amenable
for our use.}
\begin{theorem}[See Theorem 9 in \cite{CohenBoxCons}]\label{thm:bcNewton}
   Suppose that there exists a point $z_\epsilon^* = (x_\epsilon^*,
   y_\epsilon^*)$ for which $\psi_{A,r,c}(x_\epsilon^*, y_\epsilon^*) - \psi^* \leq \epsilon^2/(3n)$
   and $\norm{z_\epsilon^*}_\infty \leq B$, where $\psi^* = \min_{x,y \in
     \R^n}\psi_{A,r,c}(x,y)$. Then, there exists a Newton-type algorithm which,
   with high probability, computes an $\widehat{x}, \widehat{y}$ such that
   \begin{equation*}
  \norm{M_{A}(\widehat{x}, \widehat{y}) \bones - r}^2_2 +
  \norm{M_{A}(\widehat{x}, \widehat{y})^T \bones - c}^2_2 \leq \epsilon
   \end{equation*}
  in $\widetilde{O}\left( n^2 B \log^2\left(s_A\right)\right)$ time-- where $s_A$ is the sum of the entries in $A$.
\end{theorem}
The following parallel complexity for the Newton-type algorithm of
Theorem \ref{thm:bcNewton} is nearly trivial, but not explicitly
stated in \cite{CohenBoxCons}. Hence, we provide a proof for completeness.
\begin{theorem}\label{thm:parBcNewton}
   Suppose that there exists a point $z_\epsilon^* = (x_\epsilon^*,
   y_\epsilon^*)$ for which $\psi_{A,r,c}(x_\epsilon^*, y_\epsilon^*) - \psi^* \leq \epsilon^2/(3n)$
   and $\norm{z_\epsilon^*}_\infty \leq B$, where $\psi^* = \min_{x,y \in
     \R^n}\psi_{A,r,c}(x,y)$. Then, there exists a Newton-type algorithm which,
   with high probability, computes an $\widehat{x}, \widehat{y}$ such that
   \begin{equation*}
  \norm{M_{A}(\widehat{x}, \widehat{y}) \bones - r}^2_2 + \norm{M_{A}(\widehat{x}, \widehat{y})^T \bones - c}^2_2 \leq \epsilon
   \end{equation*}
  in $\widetilde{O}\left(n^2 B \log^2\left(s_A\right) \right)$ total work and
  $\widetilde{O}\left(B \log^2\left(s_A \right)\right)$ depth.
\end{theorem}
\begin{proof}
  From the proof of Theorem 3.4 in \cite{CohenBoxCons}, observe that the Newton-type algorithm of
  Theorem \ref{thm:bcNewton} performs $\widetilde{O}\left( B \log^2 \left(s_A \right) \right)$
  sequential (box-constrained) Newton steps on the function
  \begin{equation*}
  f(x,y) = \psi_{A,r,c}(x, y) + \frac{\epsilon^2}{36n^2e^{B}}\left( \sum_{i\in
        [n]}\left(  e^{x_i} + e^{-x_i} + e^{y_i} +e^{-y_i} \right)\right)
  \end{equation*}
  Hence, it suffices to show that each Newton iteration
  can be implemented in $\widetilde{O}\left( n^2 \right)$ total work and
  $\widetilde{O}\left(1\right)$ depth.

  From the proof of Theorem 5.11 in
  \cite{CohenBoxCons}, each Newton-step consists of constructing a vertex
  sparsifier chain $\spchain$ (see Definition 5.9 in
  \cite{LeeSDD}) for the Hessian $\nabla^2 f(x^{(k)}, y^{(k)})$ at the current Newton iterate $x^{(k)}, y^{(k)}$ and then applying the procedure \textsc{OptimizeChain} (see Figure
  5.2 in \cite{CohenBoxCons}) to $\spchain$ and the gradient
  $\nabla f(x^{(k)}, y^{(k)})$. Trivially, the Hessian and gradient of $f$ can
  be computed in $O(n^2)$ work and $\widetilde{O}(1)$ depth. Further, by Theorem
  5.10 in \cite{LeeSDD}, we know that a vertex sparsifier chain $\spchain$ of length $d = O(\log n)$ and total sparsity $O(n)$
  can be constructed for the Hessian in $O(n^2)$ work and $\widetilde{O}(1)$ depth. Thus,
  it need only be shown that \textsc{OptimizeChain} can be implemented in
  $\widetilde{O}(n^2)$ total work and $\widetilde{O}(1)$ depth.

  The procedure \textsc{OptimizeChain} applies the
  subroutines \textsc{ApproxMapping} (see Figure 5.1 in \cite{CohenBoxCons}) and
  \textsc{FastSolve} (see Lemma 5.3 in \cite{CohenBoxCons}) to the members $\spmem$ of the vertex sparsifier
  chain. The approximate voltage extension subroutine \textsc{ApproxMapping}
  computes $O\left( \log \left( 1/ \epsilon \right) \right)$ matrix-vector
  multiplications using $M^{(t)}$ and disjoint sub-matrices of $\nabla^2
  f(x^{(k)}, y^{(k)})$ induced by the vertices $F_t$. Hence,
  \textsc{ApproxMapping} can be applied to all of the $O(\log n)$ members of the vertex sparsifier in $\widetilde{O}(n^2)$ total work and $\widetilde{O}(1)$ depth.

  Further, for each $M^{(t)}$, \textsc{FastSolve} performs $O(1)$
  iterations of projected gradient descent on a quadratic function in $M^{(t)}$;
  where the projection is onto an $\ell_\infty$ ball. Since the gradient of any
  quadratic in $M^{(t)}$ can be calculated in time equal to the sparsity of $M^{(t)}$ and projection
  onto an $\ell_\infty$ ball can be implemented simply by truncating
  coordinates, it follows that \textsc{FastSolve} can be applied to all the
  members of $\spchain$ in $O(n)$ total work and $\widetilde{O}(1)$ depth. Thus,
  \textsc{OptimizeChain} can be implemented in $\widetilde{O}(n^2 )$ total work and
  $\widetilde{O}(1)$ depth.
\end{proof}
One would like to immediately apply Theorems \ref{thm:bcNewton} and
\ref{thm:parBcNewton} to give the oracles for Theorem
\ref{thm:matScalRed} and Corollary \ref{cor:parMatScalRed}. Unfortunately, there is a mismatch between the $l_1$
guarantee required by Definition \ref{def:MatScal} and the $l_2$ guarantee in
Theorem \ref{thm:bcNewton} for which we need the following lemma.
 \begin{lemma}\label{lem:2norm21}
   Suppose that there exists a point $z_\epsilon^* = (x_\epsilon^*,
   y_\epsilon^*)$ for which $\psi_{A,r,c}(x_\epsilon^*, y_\epsilon^*) - \psi^* \leq \epsilon^4/\left(3n^3\right)$
   and $\norm{z_\epsilon^*}_\infty \leq B$, where $\psi^* = \min_{x,y \in \R^n}\psi_{A,r,c}(x,y)$. Then, there exists a Newton-type algorithm which
   computes an $\widehat{x}, \widehat{y}$ such that
   \begin{equation*}
  \norm{M_{A}(\widehat{x}, \widehat{y}) \bones - r}_1 + \norm{M_{A}(\widehat{x}, \widehat{y})^T \bones - c}_1 \leq \epsilon
   \end{equation*}
  in time/total work $\widetilde{O}\left( n^2 B \log^2\left( s_A\right) \right)$ and with $\widetilde{O}\left(B \log^2\left( s_A\right) \right)$ depth.
\end{lemma}
\begin{proof}
Let $\delta = \epsilon^2/(2n)$ be the error tolerance used in Theorem
\ref{thm:bcNewton} and Theorem \ref{thm:parBcNewton}. Then, by Cauchy-Schwartz and the inequality
$(a+b)^2 \leq 2(a^2 + b^2)$ we have
\begin{align*}
  \hspace{-2em}\left(\norm{B_{C/\eta}(\widehat{x}, \widehat{y}) \bones - r}_1 + \norm{B_{C/\eta}(\widehat{x},
      \widehat{y})^T \bones - c}_1\right)^2 &\leq n\left( \norm{B_{C/\eta}(\widehat{x}, \widehat{y}) \bones - r}_2 + \norm{B_{C/\eta}(\widehat{x},
                                              \widehat{y})^T \bones - c}_2\right)^2 \\
  &\leq \epsilon^2
\end{align*}
Hence, for such a $\delta$, the algorithms of Theorems
\ref{thm:bcNewton} and \ref{thm:parBcNewton} have the same sequential and
parallel complexities, respectively, and produce a $\hat{x},\hat{y}$ satisfying
\begin{equation*}
  \norm{B_{C/\eta}(\widehat{x}, \widehat{y}) \bones - r}_1 + \norm{B_{C/\eta}(\widehat{x},
      \widehat{y})^T \bones - c}_1 \leq \epsilon
\end{equation*}
\end{proof}
The final step before combining Theorem \ref{thm:matScalRed}, Corollary
\ref{cor:parMatScalRed}, and Lemma \ref{lem:2norm21} is to bound the
constant $B$ in Lemma \ref{lem:2norm21} in terms of $\nu = \max_{i,j} 1/A_{ij}$ and $\xi= \max_{i} 1/\min(r_i,c_i)$.
 \begin{lemma}\label{lem:widthbound}
   Suppose that $A$ and $r,c$ are strictly positive in \eqref{eq:MatScalA} and
   satisfy the hypotheses of Definition \ref{def:MatScal}, then there
   exists an optimal solution $z^* = \left(x^*,y^*\right)$ such that
$\norm{z^*}_\infty \leq 2\log \left(n \nu \xi\right)$
where $\nu,\xi$ are as defined in Theorem \ref{thm:matScalRed}.
 \end{lemma}
 \begin{proof}
   From \cite{LinialMatrixScaling} and the fact that $A$ and $r,c$ are strictly
   positive, there exists an optimal solution $z^* = \left(  x^*,y^*\right)$. It is easy to
   see that for any $\alpha \in \R$, $\left(x^* + \alpha\bones, y^* - \alpha\bones\right)$
   is also optimal. Hence, without loss of generality, we can assume that $z^*$
   is an optimal solution such that $\min_{i\in[n]} \, \{x_i^*\}=0$.

   Let $m$ be such that $x_m^* = 0$. For such a $z^*$, notice that first-order optimality conditions imply that
   \begin{equation*}
   \frac{e^{\max_i \{y_i^*\}}}{\nu} \leq e^{x_m}\sum_{i\in
       [n]}e^{y_i}A_{m,i} = r_m  \leq 1 \hspace{1em} \text{and} \hspace{1em} r_m = e^{x_m}\sum_{i\in
       [n]}e^{y_i}A_{m,i} \leq ne^{\max_i \{y_i^*\}}
   \end{equation*}
   where we have used that fact that $A_{i,j},r_i \leq 1$ for all $i,j$.
   This gives that $\max_i \{y_i^*\} \leq \log\left(\nu \right)$ and $-\max_i \{y_i^*\} \leq \log\left(n\xi \right)$.
   Additionally, for $k = \argmax_i \{x_i^*\}$ and $t = \argmin_i \{y_i^*\}$ we have
   \begin{equation*}
   \frac{e^{x_k+\max_i \{y_i^*\}}}{\nu} \leq e^{x_k}\sum_{i\in
     [n]}e^{y_i}A_{k,i} = r_k  \leq 1 \hspace{1em} \text{and} \hspace{1em} c_t =e^{y_t}\sum_{i\in [n]}e^{x_i}A_{i,t} \leq ne^{y_t+x_k}
   \end{equation*}
   This yields $\max_i \{x_i^*\} \leq \log (n \nu \xi)$ and $-\min_i \{y_i^*\} \leq 2\log\left(n\nu\xi \right)$.
   Putting all these bounds together, it follows that $ \norm{z^*}_\infty\leq
   2\log \left(n\nu \xi\right)$
 \end{proof}
 Using Lemma \ref{lem:widthbound}, we can now prove our final result.
 \begin{theorem}\label{thm:finalalg}
   Consider an instance of the optimal transport problem \eqref{eq:Primal}.
   There exists an algorithm which computes an $\epsilon$-approximate solution
   with high probability in time
   \begin{equation*}
   \widetilde{O}\left( \frac{n^2\norm{C}_{\infty}}{\epsilon} \right)
   \end{equation*}
   and in parallel with $\widetilde{O}\left(n^2\norm{C}_{\infty}/\epsilon \right)$
   total work and $\widetilde{O}\left(\norm{C}_{\infty}/\epsilon \right)$
   depth.
 \end{theorem}
\begin{proof}
  Consider the Newton-type algorithm of Lemma \ref{lem:2norm21}. By Lemma
  \ref{lem:widthbound}, when $A,r,c$ are strictly positive and
  satisfy the hypotheses of the matrix scaling problem, it follows that $B = O(\log(n\nu\xi))$
  and $s_A = O(n^2)$-- where $\nu$ and $\xi$ are as defined in Theorem
  \ref{thm:matScalRed}.
  Hence, in this case, the algorithm runs in
  \begin{equation*}
  \widetilde{O}\left(n^2\log\left(n\nu \xi\right)\right)\hspace{0.1in} \text{time/total work and } \widetilde{O}\left(\log\left(n\nu \xi\right)\right) \hspace{0.1in} \text{depth}
  \end{equation*}
  This gives an oracle satisfying
  the requirements of Theorem \ref{thm:matScalRed} and Corollary
  \ref{cor:parMatScalRed} where, respectively,
  \begin{equation*}
  \widetilde{O}\left(  \mathcal{T}\left(n, \frac{1}{\epsilon}, \nu, \xi
    \right)\right) = \widetilde{O}\left(  \mathcal{T}_w\left(n,
      \frac{1}{\epsilon}, \nu, \xi \right)\right) = \widetilde{O}\left(n^2\log\left(n\nu \xi\right)\right)
  \end{equation*}
  and
  \begin{equation*}
  \widetilde{O}\left(  \mathcal{T}_d\left(n, \frac{1}{\epsilon}, \nu, \xi \right)\right) =
    \widetilde{O}\left(\log\left(n\nu \xi\right)\right)
  \end{equation*}
  Plugging in for $\nu$ and $\xi$, it follows that
  \begin{equation*}
  \widetilde{O}\left(\mathcal{T}\left(n, \frac{16\norm{C}_\infty}{\epsilon},
      n^{8\norm{C}_\infty/\epsilon}, \frac{4 \norm{C}_\infty n}{\epsilon}\right)
  \right) = \widetilde{O}\left(\frac{n^2\norm{C}_\infty}{\epsilon} \right)
  \end{equation*}
and
  \begin{equation*}
  \widetilde{O}\left(\mathcal{T}_d\left(n, \frac{16\norm{C}_\infty}{\epsilon},
      n^{8\norm{C}_\infty/\epsilon}, \frac{4 \norm{C}_\infty n}{\epsilon}\right)
  \right) = \widetilde{O}\left( \frac{\norm{C}_\infty}{\epsilon} \right)
  \end{equation*}
giving the result.
\end{proof}


\section{Hardness Reduction}\label{sec:bimatch}
 In this section, we show that the $\widetilde{O}\left( n^2 \norm{C}_\infty/\epsilon \right)$ complexity
 of the previously-derived algorithms for the optimal
 transportation problem \eqref{eq:Primal} cannot be improved without using fast matrix multiplication (i.e. \cite{IbarraM81}) barring a breakthrough on a long-standing open problem in algorithmic graph theory. Formally, we show that any
 further improvement in the complexity of solving \eqref{eq:Primal},
 would yield a $o(n^{2.5})$ algorithm for \textit{maximum cardinality bipartite
 matching}. Currently, the only known algorithms which achieve
such a complexity are based on fast matrix multiplication (i.e. \cite{IbarraM81}).

In order to prove this reduction, consider an instance of the maximum cardinality bipartite
 matching problem
where we have an undirected, bipartite graph $G=(V,E)$ such that $V$ is
the union of disjoint sets of vertices $L$ and $R$ (each of size $n$) and all edges
go exclusively between $L$ and $R$, i.e. $E\subseteq L\times R$. Our goal is
to compute a matching, $F\subseteq E$ with

\begin{equation*}
\deg_F(i) \defeq \left|\left\{  j\in V\,|\,\{i,j\}\in F\right\}\right|\leq 1 \hspace{0.5in} \forall
  i \in V
\end{equation*}
which maximizes $|F|$. Consider the following lemma
\begin{lemma}\label{lem:hardres}
  Given an oracle for computing an $\epsilon$-approximate solution to the
  optimal transportation problem \eqref{eq:Primal} (under the assumption $\norm{C}_\infty= O(1)$)
   in time $T(n,\epsilon)$, one can compute a maximum
  cardinality matching $F$ in time $O(T(n,\epsilon) +n^3\epsilon)$.
\end{lemma}
\begin{proof}
We reduce an instance of the bipartite matching problem to optimal transport as follows.
Without loss of generality, let $L=[n]$ and $R=[n]$ and let $r=c=\frac{1}{n}\bones$.
Furthermore, define a cost matrix $C\in\R^{n\times n}$ with $C_{ij}=0$
if $\{i,j\}\in E$ and $C_{ij}=1$ otherwise.

Now, suppose we solve the optimal transport problem corresponding to these
inputs to $\epsilon$-accuracy. Define $OPT_T$ to be the optimal value of this
transportation problem and let $OPT_M$ to be the optimal value of the
maximum cardinality matching in our graph.
Clearly, we have computed an $X$ with $X\bones=X^T\bones=\frac{1}{n}\bones$
and such that $\iprod{C}{X}\leq OPT_{T}+\epsilon$. Furthermore, notice that by
taking the maximum matching in our graph adding an arbitrary matching between
it's unmatched vertices, we can create a perfect matching $Y\in [0,1]^{n\times
  n}$ such that $\frac{1}{n}Y$ is feasible for our optimal transportation
problem and we have $\iprod{C}{Y} = 1 - OPT_M/n$. Hence
$\epsilon$-optimality of $X$ implies that
\begin{equation*}
\iprod{C}{X} \leq 1 + \epsilon -\frac{OPT_M}{n}
\end{equation*}
Hence, as $Z = nX$ is a fractional perfect matching in our graph, this result
immediately implies that our oracle for solving optimal transport gives us a
fractional perfect matching $Z$ where $\iprod{C}{Z} \leq (1+\epsilon)n - OPT_M$.
By removing all flow in $Z$ along edges $(i,j)$ where $C_{ij} = 1$ (i.e. edges
which are non-existent in our original graph) and then rounding the
corresponding fractional matching to an actual matching \cite{KangP15} (which can be done in
nearly-linear time) we obtain an actual matching $\hat{Z}$ such
that
\begin{equation*}
\iprod{C}{\hat{Z}} \leq (1+\epsilon)n - OPT_M
\end{equation*}
Hence, $\hat{Z}$ is a matching which has at least $OPT_M -n\epsilon $ edges.
Thus, by running augmenting paths \cite{AugPaths} on $\hat{Z}$ in $O(n^3\epsilon)$ time (since
$G$ is dense)
we can find the remaining $n\epsilon$ edges in the maximum matching. This yields an algorithm with complexity
\begin{equation*}
O \left( T(n,\epsilon) + n^3\epsilon \right)
\end{equation*}
for finding a maximum matching in a dense graph.
\end{proof}

Using Lemma \ref{lem:hardres}, we see that, if $T(n,\epsilon)=\widetilde{O}(n^{2}/\epsilon)$, picking $\epsilon=1/\sqrt{n}$
gives a $\widetilde{O}\left(n^{2.5}\right)$ algorithm for matching. For any
smaller $T(n,\epsilon)$ (more than log factors of course) an
appropriate choice of $\epsilon$ would give a $o(n^{2.5})$ algorithm for maximum
cardinality bipartite matching.


\bibliographystyle{alpha}
\bibliography{refs}

\newcommand{\etalchar}[1]{$^{#1}$}
\begin{thebibliography}{BvdPPH11}

\bibitem[ACB17]{UL1}
Martin Arjovsky, Soumith Chintala, and L{\'e}on Bottou.
\newblock {W}asserstein generative adversarial networks.
\newblock In Doina Precup and Yee~Whye Teh, editors, {\em Proceedings of the
  34th International Conference on Machine Learning}, volume~70 of {\em
  Proceedings of Machine Learning Research}, pages 214--223, International
  Convention Centre, Sydney, Australia, 06--11 Aug 2017. PMLR.

\bibitem[ANOY14]{MoreC2}
Alexandr Andoni, Aleksandar Nikolov, Krzysztof Onak, and Grigory Yaroslavtsev.
\newblock Parallel algorithms for geometric graph problems.
\newblock In {\em Proceedings of the Forty-sixth Annual ACM Symposium on Theory
  of Computing}, STOC '14, pages 574--583, New York, NY, USA, 2014. ACM.

\bibitem[AS14]{MoreC1}
Pankaj~K. Agarwal and R.~Sharathkumar.
\newblock Approximation algorithms for bipartite matching with metric and
  geometric costs.
\newblock In {\em Proceedings of the Forty-sixth Annual ACM Symposium on Theory
  of Computing}, STOC '14, pages 555--564, New York, NY, USA, 2014. ACM.

\bibitem[AWR17]{WeedSinkhorn}
Jason Altschuler, Jonathan Weed, and Philippe Rigollet.
\newblock Near-linear time approximation algorithms for optimal transport via
  sinkhorn iteration.
\newblock In I.~Guyon, U.~V. Luxburg, S.~Bengio, H.~Wallach, R.~Fergus,
  S.~Vishwanathan, and R.~Garnett, editors, {\em Advances in Neural Information
  Processing Systems 30}, pages 1964--1974. Curran Associates, Inc., 2017.

\bibitem[AZLOW17]{OtherScaling}
Z.~Allen-Zhu, Y.~Li, R.~Oliveira, and A.~Wigderson.
\newblock Much faster algorithms for matrix scaling.
\newblock In {\em 2017 IEEE 58th Annual Symposium on Foundations of Computer
  Science (FOCS)}, pages 890--901, Oct 2017.

\bibitem[AZO18]{OrechPacking}
Zeyuan Allen-Zhu and Lorenzo Orecchia.
\newblock Nearly linear-time packing and covering lp solvers.
\newblock {\em Mathematical Programming}, Feb 2018.

\bibitem[BK17]{DRO2}
Jose Blanchet and Yang Kang.
\newblock Distributionally robust groupwise regularization estimator.
\newblock In Min-Ling Zhang and Yung-Kyun Noh, editors, {\em Proceedings of the
  Ninth Asian Conference on Machine Learning}, volume~77 of {\em Proceedings of
  Machine Learning Research}, pages 97--112. PMLR, 15--17 Nov 2017.

\bibitem[BKZM17]{DRO1}
Jose Blanchet, Yang Kang, Fan Zhang, and Karthyek Murthy.
\newblock Data-driven optimal transport cost selection for distributionally
  robust optimizatio.
\newblock 05 2017.

\bibitem[BvdPPH11]{CV2}
Nicolas Bonneel, Michiel van~de Panne, Sylvain Paris, and Wolfgang Heidrich.
\newblock Displacement interpolation using lagrangian mass transport.
\newblock {\em ACM Trans. Graph.}, 30(6):158:1--158:12, December 2011.

\bibitem[CK18]{SimplerSink}
Deeparnab Chakrabarty and Sanjeev Khanna.
\newblock {Better and Simpler Error Analysis of the Sinkhorn-Knopp Algorithm
  for Matrix Scaling}.
\newblock In Raimund Seidel, editor, {\em 1st Symposium on Simplicity in
  Algorithms (SOSA 2018)}, volume~61 of {\em OpenAccess Series in Informatics
  (OASIcs)}, pages 4:1--4:11, Dagstuhl, Germany, 2018. Schloss
  Dagstuhl--Leibniz-Zentrum fuer Informatik.

\bibitem[CMTV17]{CohenBoxCons}
Michael~B. Cohen, Aleksander Madry, Dimitris Tsipras, and Adrian Vladu.
\newblock Matrix scaling and balancing via box constrained newton’s method
  and interior point methods.
\newblock {\em 2017 IEEE 58th Annual Symposium on Foundations of Computer
  Science (FOCS)}, Oct 2017.

\bibitem[Cut13]{CuturiSinkhorn}
Marco Cuturi.
\newblock Sinkhorn distances: Lightspeed computation of optimal transport.
\newblock In {\em Proceedings of the 26th International Conference on Neural
  Information Processing Systems - Volume 2}, NIPS'13, pages 2292--2300, USA,
  2013. Curran Associates Inc.

\bibitem[DGK18]{BetterSinkAnaly}
Pavel Dvurechensky, Alexander Gasnikov, and Alexey Kroshnin.
\newblock Computational optimal transport: Complexity by accelerated gradient
  descent is better than by sinkhorn’s algorithm.
\newblock In Jennifer Dy and Andreas Krause, editors, {\em Proceedings of the
  35th International Conference on Machine Learning}, volume~80 of {\em
  Proceedings of Machine Learning Research}, pages 1366--1375,
  Stockholmsmässan, Stockholm Sweden, 10--15 Jul 2018. PMLR.

\bibitem[Ful61]{AugPaths}
D.~Fulkerson.
\newblock An out-of-kilter method for minimal-cost flow problems.
\newblock {\em Journal of the Society for Industrial and Applied Mathematics},
  9(1):18--27, 1961.

\bibitem[GCPB16]{BachSGD}
Aude Genevay, Marco Cuturi, Gabriel Peyr\'{e}, and Francis Bach.
\newblock Stochastic optimization for large-scale optimal transport.
\newblock In D.~D. Lee, M.~Sugiyama, U.~V. Luxburg, I.~Guyon, and R.~Garnett,
  editors, {\em Advances in Neural Information Processing Systems 29}, pages
  3440--3448. Curran Associates, Inc., 2016.

\bibitem[IM81]{IbarraM81}
Oscar~H. Ibarra and Shlomo Moran.
\newblock Deterministic and probabilistic algorithms for maximum bipartite
  matching via fast matrix multiplication.
\newblock {\em Inf. Process. Lett.}, 13(1):12--15, 1981.

\bibitem[Kan58]{KantorovitchProb}
L.~Kantorovitch.
\newblock On the translocation of masses.
\newblock {\em Management Science}, 5(1):1--4, 1958.

\bibitem[KP15]{KangP15}
Donggu Kang and James Payor.
\newblock Flow rounding.
\newblock {\em CoRR}, abs/1507.08139, 2015.

\bibitem[LPS15]{LeeSDD}
Yin~Tat Lee, Richard Peng, and Daniel~A. Spielman.
\newblock Sparsified cholesky solvers for sdd linear systems, 2015.

\bibitem[LS14]{SidLeeSolv}
Y.~T. Lee and A.~Sidford.
\newblock Path finding methods for linear programming: Solving linear programs
  in square root rank iterations and faster algorithms for maximum flow.
\newblock In {\em 2014 IEEE 55th Annual Symposium on Foundations of Computer
  Science}, pages 424--433, Oct 2014.

\bibitem[LS15]{SidInvMain}
Yin~Tat Lee and Aaron Sidford.
\newblock Efficient inverse maintenance and faster algorithms for linear
  programming.
\newblock In {\em Proceedings of the 2015 IEEE 56th Annual Symposium on
  Foundations of Computer Science (FOCS)}, FOCS '15, pages 230--249,
  Washington, DC, USA, 2015. IEEE Computer Society.

\bibitem[LSW98]{LinialMatrixScaling}
Nathan Linial, Alex Samorodnitsky, and Avi Wigderson.
\newblock A deterministic strongly polynomial algorithm for matrix scaling and
  approximate permanents.
\newblock In {\em Proceedings of the Thirtieth Annual ACM Symposium on Theory
  of Computing}, STOC '98, pages 644--652, New York, NY, USA, 1998. ACM.

\bibitem[MEK18]{DRO3}
Peyman Mohajerin~Esfahani and Daniel Kuhn.
\newblock Data-driven distributionally robust optimization using the
  wasserstein metric: performance guarantees and tractable reformulations.
\newblock {\em Mathematical Programming}, 171(1):115--166, Sep 2018.

\bibitem[PZ16]{OTStats2}
Victor~M. Panaretos and Yoav Zemel.
\newblock Amplitude and phase variation of point processes.
\newblock {\em Ann. Statist.}, 44(2):771--812, 04 2016.

\bibitem[Qua18]{ConcurrentPaper}
Kent Quanrud.
\newblock {Approximating Optimal Transport With Linear Programs}.
\newblock In Jeremy~T. Fineman and Michael Mitzenmacher, editors, {\em 2nd
  Symposium on Simplicity in Algorithms (SOSA 2019)}, volume~69 of {\em
  OpenAccess Series in Informatics (OASIcs)}, pages 6:1--6:9, Dagstuhl,
  Germany, 2018. Schloss Dagstuhl--Leibniz-Zentrum fuer Informatik.

\bibitem[SA12]{MoreC3}
R.~Sharathkumar and Pankaj~K. Agarwal.
\newblock A near-linear time epsilon-approximation algorithm for geometric
  bipartite matching.
\newblock In {\em Proceedings of the Forty-fourth Annual ACM Symposium on
  Theory of Computing}, STOC '12, pages 385--394, New York, NY, USA, 2012. ACM.

\bibitem[SdGP{\etalchar{+}}15]{CV1}
Justin Solomon, Fernando de~Goes, Gabriel Peyr{\'e}, Marco Cuturi, Adrian
  Butscher, Andy Nguyen, Tao Du, and Leonidas Guibas.
\newblock Convolutional wasserstein distances: Efficient optimal transportation
  on geometric domains.
\newblock {\em ACM Trans. Graph.}, 34(4):66:1--66:11, July 2015.

\bibitem[SR]{OTStats1}
Gábor~J. Székely and Maria~L. Rizzo.
\newblock Testing for equal distributions in high dimensions.
\newblock {\em InterStat}, page 2004.

\end{thebibliography}

\end{document}